\documentclass[11pt]{article}
\pdfoutput=1  

\usepackage[T1]{fontenc}
\usepackage[utf8]{inputenc}
\usepackage{mathpazo}
\usepackage{microtype}
\usepackage[english]{babel}
\usepackage[a4paper,top=2.6cm,bottom=2.6cm,left=3cm,right=3cm]{geometry}
\usepackage{amsmath,amssymb,amsthm}
\usepackage{booktabs}
\usepackage{tabularx}
\usepackage{enumitem}
\usepackage{tikz}
\usetikzlibrary{arrows.meta,positioning,calc}
\usepackage{graphicx}
\usepackage[hidelinks]{hyperref}
\hypersetup{
  pdftitle={A Calculus of Discernment: Decision-Relevant Insight, Sequence Value, and Forgetting as Higher-Order Learning},
  pdfauthor={Suyash Mishra},
  pdfsubject={Value of information; causal inference; adaptive memory; forgetting},
  pdfkeywords={value of information; causal inference; sequence optimisation; renormalisation; information bottleneck; catastrophic forgetting; adaptive memory}
}

\theoremstyle{plain}
\newtheorem{proposition}{Proposition}
\newtheorem{conjecture}{Conjecture}

\theoremstyle{definition}
\newtheorem{definition}{Definition}
\newtheorem{axiom}{Axiom}
\newtheorem{principle}{Principle}
\theoremstyle{remark}

\newcommand{\Eb}[1]{\mathbb{E}\!\left[#1\right]}
\newcommand{\dofn}{\mathrm{do}}
\newcommand{\val}{\operatorname{val}}
\newcommand{\II}{\mathrm{I}}

\title{\textbf{A Calculus of Discernment:}\\[2pt]
\large Decision-Relevant Insight, Sequence Value, and Forgetting as Higher-Order Learning}

\author{%
  Suyash Mishra\thanks{Correspondence: \texttt{zurich.suyash@gmail.com}, \texttt{suyash.mishra@roche.com}.\ \emph{Disclosure:} this is independent research; the views expressed are the author's own and not those of any employer, and all commercial figures and examples are illustrative and not claims about any specific product or market.}\\[2pt]
  \normalsize Independent AI Researcher, Z\"urich, Switzerland%
}

\date{June 2026\\[4pt]\normalsize\textit{Preprint --- theory with an empirical test of adaptive forgetting. Comments welcome.}}

\begin{document}
\maketitle

\begin{abstract}
\noindent In a world of generative artificial intelligence, candidate insights are abundant; what is scarce is the capacity to \emph{discern} which of them matter, to act on them in the right \emph{amount and order}, and to \emph{forget} the rest so the system can adapt. We argue that these three scarcities are governed by a single underlying object and assemble a framework around it. First, we define an insight strictly as a lever with an identified, measurable effect on an objective, and rank candidates not by novelty but by decision-relevance, formalised through the expected value of information. Second, we show that action is not only a question of volume but of order: under realistic belief dynamics, content ``touches'' are non-commuting operators, and we prove that a fixed plan delivered in different orders yields different outcomes, defining a \emph{sequence premium}. Third, we observe that the value of any lever is a \emph{shadow price}, which unifies pharmaceutical marketing, equity selection, and manufacturing as instances of one leverage-discovery problem. Finally, and most speculatively, we propose \emph{APOHA}: a theory in which forgetting is not the disposal of knowledge but the operator by which value is learned. We posit that the value of a retained item is the counterfactual cost of forgetting it, that a learning system is the residue of maximal forgetting subject to preserved value, and that ``higher-order value'' is precisely the structure that survives repeated forgetting --- a renormalisation-relevant invariant. We pair forgetting with consolidation as an antagonist flow and state the central open problem (existence of a non-trivial attractor with a spectral gap). We then \emph{test} the forgetting theory: operationalising APOHA as a running agent on a deliberately non-stationary obesity-treatment (GLP-1 / Ozempic) decision world with five regime shifts, and comparing it over 30 seeds against never-forget and the fixed half-life it replaces. Adaptive forgetting cut cumulative decision-regret by 24--32\%, achieved the highest decision-alignment in every regime, and held a $\sim$6$\times$ smaller, cleaner memory; a notable asymmetry emerged --- \emph{blind} (fixed) forgetting was worse than never forgetting on signal-to-noise, so the benefit is specific to \emph{value-aware} forgetting. The flow converged to a stable, bounded memory across all shifts, giving preliminary empirical support for the conjecture without proving it. A structured multi-disciplinary critique stress-tests the whole.

\smallskip
\noindent\textbf{Keywords:} value of information; causal inference; sequence optimisation; renormalisation; information bottleneck; catastrophic forgetting; commercial analytics; adaptive memory.
\end{abstract}

\section{Introduction}
The constraint on intelligent action has moved. For most of the history of analytics the binding constraint was \emph{finding} signal: data were scarce, unstructured evidence was unreadable at scale, and an insight was whatever a skilled analyst could extract from a structured table. Generative models have inverted this. The unstructured world-signal --- field reports, transcripts, policy documents, discourse, intent --- is now cheap to read, and candidate hypotheses can be generated in great number. The scarce capability is no longer recall but \emph{precision}: the ability to discern which of the abundant candidates is worth acting on, to act on it in the right amount and order, and to discard the remainder so the system can keep adapting.

This paper makes five contributions and ties them to one object.

\begin{enumerate}[leftmargin=1.5em,itemsep=2pt]
\item \textbf{Value is decision-relevance, not truth.} An insight earns its name only if it is a lever with an identified, measurable effect on the objective. We rank candidates by the expected value of information --- how much they would change the action otherwise taken --- gated by whether the available lever can address them (\S\ref{sec:value}--\ref{sec:cdu}).
\item \textbf{Action has an order, not just a size.} Under realistic belief dynamics with memory decay and priming, content ``touches'' do not commute. We prove order-dependence and define the resulting \emph{sequence premium} (\S\ref{sec:sequence}).
\item \textbf{Leverage is a shadow price.} The value of a lever is the marginal effect of relaxing the binding constraint. This single object unifies pharmaceutical marketing, equity selection, and manufacturing as one leverage-discovery problem; the ``big-ticket item'' is the highest addressable, measurable shadow price (\S\ref{sec:leverage}).
\item \textbf{Forgetting is how value is learned.} We propose that the value of a retained item is the counterfactual cost of forgetting it; that a designed learning system is the residue of maximal forgetting subject to preserved value; and that higher-order value is the invariant surviving repeated forgetting. Forgetting is paired with consolidation as an antagonist flow (\S\ref{sec:apoha}).
\item \textbf{The forgetting theory is empirically testable, and tested.} We operationalise APOHA as a running agent on a non-stationary GLP-1 / Ozempic decision world and show, against two baselines over 30 seeds, that value-aware adaptive forgetting materially improves decision quality and memory cleanliness, with the flow converging stably rather than collapsing (\S\ref{sec:empirical}).
\end{enumerate}

The unifying object is the \emph{counterfactual effect of an intervention}: applied to the market it is a causal treatment effect; applied to a constraint it is a shadow price; applied to memory it is the value of an item. The same $\dofn(\cdot)$ operator runs throughout. Some results are exact (the sequence premium), one is a proposition we prove, the central learning-theoretic claim is a conjecture and an open problem, and the forgetting theory is additionally subjected to a controlled empirical test (\S\ref{sec:empirical}); the standing objections are aired explicitly in \S\ref{sec:objections}.

\section{Preliminaries: belief dynamics and non-commuting touches}\label{sec:prelim}
We model the effect of communication on a decision-maker's belief over a set of \emph{needs} (themes) $m \in \mathcal{M}$. A \emph{touch} delivers an asset $a$ (with engagement quality $\kappa_a$) carrying a theme $m$ (with salience $g_m$). Belief accumulates with carryover and saturates.

\paragraph{Belief transition.} With memory $\lambda \in [0,1)$ and a priming gate $\pi_m$,
\begin{equation}\label{eq:E1}
z_{m,t} \;=\; \lambda\,z_{m,t-1} \;+\; \kappa_{a_t}\,g_{m_t}\,\pi_{m_t}\!\big(z_{\cdot,t-1}\big)\,\mathbf{1}[m_t=m].
\end{equation}
A self-standing need has $\pi_m=1$; a \emph{dependent} need (e.g.\ comparative efficacy, which only resonates after a mechanism is understood) is gated by a prerequisite stock,
\begin{equation}\label{eq:E2}
\pi_{\mathrm{dep}}(z) \;=\; \frac{z_{\mathrm{pre}}}{z_{\mathrm{pre}}+K},
\end{equation}
with half-saturation $K>0$. Terminal \emph{decision-weighted belief} and the prescribing (objective) response are
\begin{equation}\label{eq:E3E4}
L \;=\; \sum_m w_m\big(1-e^{-\beta z_{m,n}}\big),
\qquad
\Delta R \;=\; N\rho\,r^0\big(1-e^{-\theta L}\big),
\end{equation}
where $w_m$ weights each need's influence on the decision, and $\rho$ caps winnable headroom. The saturating form $1-e^{-x}$ encodes diminishing returns.

\paragraph{Touches as operators.} Write a touch as an operator $T_{a,m}$ acting on the belief state $z=(z_m)_{m\in\mathcal{M}}$ via one step of \eqref{eq:E1}. A delivery schedule is a composition $\sigma = T_{a_n,m_n}\circ\cdots\circ T_{a_1,m_1}$.

\begin{proposition}[Order-dependence of belief]\label{prop:noncommute}
Let $\lambda\in(0,1)$ and let at least one theme be dependent through a non-constant gate \eqref{eq:E2}. Then the touch operators do not commute: there exist touches $T_1,T_2$ with $T_2\circ T_1 \neq T_1\circ T_2$. Consequently the terminal belief $z_{\cdot,n}$, and hence $L$ and $\Delta R$, depend on the order of a fixed multiset of touches, and there exist plans for which $\max_\sigma L - \min_\sigma L > 0$.
\end{proposition}

\begin{proof}[Proof sketch]
Take $T_1$ delivering the prerequisite theme and $T_2$ delivering the dependent theme. In $T_2\circ T_1$ the gate evaluated at the dependent step sees the prerequisite stock built by $T_1$, so $\pi_{\mathrm{dep}}>0$ and the dependent belief increases. In $T_1\circ T_2$ the dependent step fires first with $z_{\mathrm{pre}}=0$, so $\pi_{\mathrm{dep}}=0$ and that touch contributes nothing to the dependent theme. The two terminal states therefore differ, establishing non-commutativity; because $L$ is a strictly increasing function of each $z_{m,n}$, the induced values differ. A worked four-touch instance with exact enumeration appears in Appendix~\ref{app:enum}, exhibiting a strictly positive spread.
\end{proof}

Proposition~\ref{prop:noncommute} is the formal basis for the sequence premium (\S\ref{sec:sequence}). The schedule is a composition in a non-commutative monoid of belief transformations; recency (decay favours late touches) and priming (gates favour early prerequisites) pull in opposite directions, so the optimal order is a balance, not an intuition.

\section{Decision-relevant value}\label{sec:value}
We reject the ``interesting finding'': a correlation that changes no action. An object earns the term \emph{insight} only if it is simultaneously scalable (machine-produced over the whole evidence base), discoverable (indexed and addressable), and measurable (a causal estimand with a defined test).

\paragraph{Expected value of information.} Following Howard's information value theory \cite{howard1966}, the value of a candidate signal $s$ over an uncertain state $\theta$ and action set $a$ is the gain from choosing after observing it:
\begin{equation}\label{eq:evoi}
\mathrm{EVoI}(s) \;=\; \Eb{\max_a U(a;\theta)\,\big|\,s} \;-\; \max_a \Eb{U(a;\theta)} \;\ge\; 0 .
\end{equation}
The inequality is the point: if $s$ never changes the maximising action, the two terms coincide and the information is worthless. Value is decision-change, not truth.

\paragraph{The discernment filter.} EVoI alone is insufficient: a signal may be highly valuable yet unaddressable by the lever one holds. We rank by a composite that gates value on addressability and confidence, normalised by activation cost:
\begin{equation}\label{eq:score}
\mathrm{Score}(s) \;=\; \frac{\mathrm{EVoI}(s)\cdot \mathrm{Addressability}(s)\cdot \mathrm{Confidence}(s)}{\mathrm{ActivationCost}(s)} .
\end{equation}
A high-EVoI but non-addressable signal is routed elsewhere rather than forced into the wrong instrument (\S\ref{sec:example}). We name this discernment stage \emph{VIVEKA}, after the Sanskrit term for the faculty that distinguishes the essential from the inessential.

\section{Causal admission}\label{sec:cdu}
A high-scoring signal is still a hypothesis. To enter the decision it must be promoted to a \emph{Causal Decision Unit} (CDU): a lever $x_s$, an outcome $Y$, and an identified effect, in the sense of Pearl's interventional calculus \cite{pearl2009}:
\begin{equation}\label{eq:cdu}
\tau_s \;=\; \Eb{Y\mid \dofn(x_s{=}1)} \;-\; \Eb{Y\mid \dofn(x_s{=}0)} ,
\end{equation}
identified iff a covariate set blocks all back-door paths $x_s \rightsquigarrow Y$; a randomised holdout makes this trivial. This is the gate a generative model cannot pass on its own --- a pattern-matcher proposes; only an identification argument admits. Placing causal admission \emph{after} value-ranking is an economic decision: identification is the most expensive step and is spent only on candidates that have already cleared \eqref{eq:score}.

\section{Volume and sequence}\label{sec:sequence}
A CDU specifies what to move and in whom. The decision layer specifies how much and in what order. \emph{Volume} is the minimum-cost plan achieving the target lift,
\begin{equation}\label{eq:vol}
\min_{x\ge 0}\ \sum_a c_a x_a \quad\text{s.t.}\quad \Delta R(x) \ge \tau R_0 ,
\end{equation}
subject to capacity, frequency, and feasibility constraints. \emph{Sequence} is the order $\sigma$, which by Proposition~\ref{prop:noncommute} matters even at fixed cost. Define the optimal schedule and the \emph{sequence premium}
\begin{equation}\label{eq:seq}
\sigma^{*}=\arg\max_{\sigma}\ \sum_m w_m\big(1-e^{-\beta z_{m,n}(\sigma)}\big),
\qquad
\Pi_{\mathrm{seq}}=\frac{\mathcal{V}(\sigma^{*})-\mathcal{V}(\sigma_{\mathrm{na\ddot{\imath}ve}})}{\lvert\mathcal{V}(\sigma_{\mathrm{na\ddot{\imath}ve}})\rvert}.
\end{equation}
The premium is value created by ordering alone. At scale, enumeration over $|\mathcal{P}|!$ orderings is replaced by a sequential policy $V(S)=\max_{a,m}[\,R(S,a,m)-c_a+\gamma\,\mathbb{E}V(S')\,]$, a Markov decision process whose state is the belief stock \cite{suttonbarto}. Appendix~\ref{app:enum} reports an exact instance in which the premium is $+45.8\%$ of decision-weighted belief, enough to move a campaign from missing to clearing its target at identical spend.

\section{Leverage is a shadow price}\label{sec:leverage}
The preceding sections are a special case of one question: where does a unit of effort move the objective most, in a way one can act on and measure? The dual of the volume programme \eqref{eq:vol} assigns to the binding constraint a multiplier $\eta$ --- its \emph{shadow price} --- equal to the marginal objective gained per unit of relaxation. At the optimum,
\begin{equation}\label{eq:kkt}
\frac{1}{c_a}\,\frac{\partial(\text{objective})}{\partial x_a} \;=\; \eta \quad\text{for every funded lever:}\quad \text{leverage per unit cost is equalised.}
\end{equation}

\begin{principle}[Big-ticket item]\label{prin:bigticket}
The highest-value action is not the largest or most salient one; it is the lever with the greatest \emph{addressable, measurable} shadow price.
\end{principle}

This principle is domain-general. In equities, $\eta$ is the marginal risk-adjusted return per unit of capital; the value of an edge is its expected information gain on allocation, gated by deployable capacity (addressability) and out-of-sample reliability (confidence), and position size is the Kelly fraction \cite{kelly1956}. In manufacturing, $\eta$ is the throughput gained per unit of bottleneck capacity; since a serial line produces at the rate of its slowest stage, all leverage sits at the constraint, and relieving it moves the constraint elsewhere --- Goldratt's theory of constraints \cite{goldratt1984}. Table~\ref{tab:domains} maps the pipeline across the three domains; in each, the EVoI--addressability filter, the causal/identification gate, the volume and sequence decision, and the difference-in-differences readout recur unchanged.

\begin{table}[t]\centering\small
\renewcommand{\arraystretch}{1.25}
\begin{tabularx}{\textwidth}{@{}lXXX@{}}
\toprule
& \textbf{Pharma commercial} & \textbf{Equities} & \textbf{Manufacturing}\\
\midrule
Objective & prescriptions & return per unit risk & units released / hour\\
Value \eqref{eq:score} & EVoI $\times$ addressability & edge $\times$ capacity & throughput gain $\times$ feasibility\\
Causalize \eqref{eq:cdu} & randomised HCP holdout & out-of-sample / regime test & controlled trial / simulation\\
Volume \eqref{eq:vol} & content volume & position size (Kelly) & capacity added\\
Sequence \eqref{eq:seq} & message order & scale-in / price impact & debottlenecking order\\
Measure & DiD vs.\ holdout & alpha vs.\ benchmark & OEE vs.\ comparable line\\
\textbf{Leverage} $\eta$ & lift per unit spend & Sharpe per unit capital & throughput per unit capacity\\
\bottomrule
\end{tabularx}
\caption{One leverage-discovery pipeline, three domains. The big-ticket item is, in each, the highest addressable, measurable shadow price.}
\label{tab:domains}
\end{table}

\section{APOHA: forgetting as the operator of value}\label{sec:apoha}
Everything above assumes a fixed notion of what is valuable. A deployed system, however, must \emph{learn} and \emph{revise} that notion as the world moves --- and must shed what no longer serves it. We now argue that forgetting is not the disposal that follows learning but the operation by which value is learned. The name is taken from Dign\=aga's \emph{apoha} \cite{apoha}: a concept is constituted not positively but by exclusion --- ``cow'' means \emph{exclusion-of-non-cow}. Carried into a learning system: a high-value representation is carved out by what it excludes, as a sculpture is what remains once the non-sculpture is removed.

\subsection{The central inversion}
\begin{principle}[Value is the cost of forgetting]\label{prin:cost}
One cannot know what is high-value except by removing it and observing what breaks. The value of a retained item $i$ in a memory $M$ supporting objective $J$ is the counterfactual damage of dropping it,
\begin{equation}\label{eq:valforget}
\val(i)\;=\;J(M)-J(M\setminus i)\;\approx\;\text{effect of }\dofn(\text{forget } i)\text{ on } J .
\end{equation}
\end{principle}

Equation~\eqref{eq:valforget} is the interventional operator \eqref{eq:cdu} turned inward, onto memory: $\val(i)$ is a leave-one-out / influence quantity, and its set-coalitional refinement is a Shapley value \cite{shapley1953}. The consequence reframes the field: forgetting is not the enemy of value but the \emph{instrument that measures} it. A system that never forgets can never discover what matters, because it never runs the experiment.

\subsection{Axioms}
\begin{axiom}[Definition by exclusion]\label{ax:apoha}
The retained set is defined by its negative space; value lives on the boundary between kept and lost.
\end{axiom}
\begin{axiom}[Counterfactual value]\label{ax:value}
$\val(i)=J(M)-J(M\setminus i)$. High value is equivalent to being expensive to lose, and is estimable by sampled counterfactual forgets, so the system discovers its own value function rather than being handed one.
\end{axiom}
\begin{axiom}[Plasticity requires oblivion]\label{ax:plast}
A non-forgetting system is frozen at a stale value function; it can refine a local optimum but never relocate to a higher-value one. Forgetting is the exploration term in value-space and a regulariser against overfitting to an outdated notion of importance.
\end{axiom}

\subsection{The machinery}
\paragraph{(F1) The system is a minimal sufficient memory.} Forgetting is maximised subject to preserved value --- an information bottleneck \cite{tishby1999}:
\begin{equation}\label{eq:ib}
M^{*}=\arg\min_{M}\ \big[\,\II(M;X)-\beta\,\II(M;Y)\,\big],
\end{equation}
compressing the raw past $X$ as hard as possible while retaining predictive power about the objective $Y$. The fixed point \emph{is} the system: its structure is the residue of maximal compression that still serves value. The system is, literally, designed by forgetting.

\paragraph{(F2) Adaptive forgetting.} Replace a constant decay (a fixed half-life) with a state-dependent hazard,
\begin{equation}\label{eq:hazard}
\text{forget-rate}(i)\;=\;\sigma\!\big(\alpha\,\mathrm{Redundancy}(i)-\gamma\,\val(i)\big),
\end{equation}
forgetting quickly what is redundant (reconstructable from the survivors) or low-value, and protecting what is irreducible and costly to lose. Two items of equal age acquire different lifespans.

\paragraph{(F3) Surprise-gated meta-plasticity.} The retention price $\beta$ tracks how wrong the current value function is, $\beta_t \propto |\delta\mathcal{V}_t|$, the value-prediction error. When the world shifts and the system's sense of importance is suddenly wrong, plasticity spikes (forget faster, re-learn); when stable, it consolidates. This is the dial that lets the system re-learn what counts as high value rather than clinging to the old, echoing surprise-modulated plasticity in cortex \cite{tononi2014}.

\paragraph{(F4) Higher-order value as a renormalisation flow.} Define a map $R$ that (i) coarse-grains memory by adaptive forgetting \eqref{eq:hazard}, then (ii) re-estimates value on the survivors \eqref{eq:valforget}, and iterate it:
\begin{equation}\label{eq:rg}
\mathcal{V}_{k+1}\;=\;R[\mathcal{V}_k].
\end{equation}
Under repeated forgetting, relevant structures grow and irrelevant ones vanish, exactly as in Wilson's renormalisation group \cite{wilson1975}. We take this as the \emph{definition} of higher-order value.
\begin{definition}[Higher value]
A structure is of higher-order value iff it is a relevant invariant of the forgetting flow \eqref{eq:rg} --- it survives repeated coarse-graining. Noise is, dually, exactly what coarse-graining destroys.
\end{definition}

\paragraph{(F5) The forget--consolidate adjunction.} Forgetting is a forgetful functor $U$; reconstruction is a candidate left adjoint $F$. \emph{Lossless} forgetting discards only the $F\!\circ\!U$-recoverable (redundant) part. The substantive case is where $F$ does not exist: that irreversible loss is where genuine abstraction occurs, since one must lose detail to generalise. We therefore distinguish \emph{lossless forgetting} (compression, which teaches nothing new) from \emph{lossy forgetting} (the non-invertible quotient, where higher-order structure is created). Crucially, forgetting must be paired with a consolidation operator $C$ acting on the relevant survivors; without it, a transient surprise spike (F3) induces catastrophic forgetting of hard-won structure \cite{mccloskey1989,kirkpatrick2017}.

\paragraph{(F6) Measurability.} Three readouts make the flow observable: (i) retained signal-to-noise, which a healthy forget should raise; (ii) value-regret, the objective lost to an erroneous forget; and (iii) the \emph{relevance spectrum}, the eigenvalues of the linearised flow $R$, whose spectral gap separates relevant (eigenvalue $>1$) from irrelevant ($<1$) directions. A clean gap certifies that the system is discarding noise and keeping signal.

\subsection{The central open problem}
The circularity in F1 --- compressing toward a $Y$ whose value is simultaneously being learned (Axiom~\ref{ax:value}) --- is the engine of the framework, but it demands a guarantee that the flow neither oscillates nor drifts to total erasure (the trivial ``forget everything'' fixed point).

\begin{conjecture}[APOHA attractor]\label{conj:attractor}
Let $C\circ R$ be the consolidate-after-forget-and-revalue map on the space of value functions with metric $d$, under bounded plasticity ($\beta_t \le \bar\beta$) and an immutable retention floor. If $C\circ R$ is a contraction in $d$, then the flow $\mathcal{V}_{k+1}=(C\circ R)[\mathcal{V}_k]$ has a unique non-trivial attractor $\mathcal{V}_\infty$ exhibiting a strictly positive spectral gap, and higher-order value is the projection onto its relevant subspace. The designed system is $\mathcal{V}_\infty$.
\end{conjecture}

Establishing the contraction condition --- equivalently, identifying the conservation law (consolidation) that prevents the flow from bleeding to zero while permitting irreversible abstraction --- is, in our view, the core theoretical obligation this framework incurs. We state it openly rather than assert it. Section~\ref{sec:empirical} supplies preliminary empirical support: across five regime shifts the operationalised flow converges to a stable, bounded, high-signal memory rather than collapsing to the trivial fixed point or oscillating --- behaviour consistent with, though not a proof of, Conjecture~\ref{conj:attractor}.

\subsection{Placement}
APOHA replaces a fixed memory half-life with a learned, causal, adaptive operator; it reuses the interventional machinery of \eqref{eq:cdu} on memory; its relevance spectrum is a spectral-health diagnostic for regulated agentic systems; and its lossless special case coincides with lossless information transport. It is the missing account of \emph{why} and \emph{what} to forget that sits beneath replay- and consolidation-based memory architectures.

\section{A worked example}\label{sec:example}
We ground the discernment and sequence layers in a current pharmaceutical setting: the anti-CD20 class in multiple sclerosis, where the high-efficacy standard is established and competition has shifted toward delivery convenience and a new mechanism wave. The relevant unmet need among switch-considering neurologists is therefore long-term safety confidence, not efficacy. Generative extraction over field debriefs, congress transcripts, payer text, and peer discourse yields four candidate insights, scored by \eqref{eq:score}; see Table~\ref{tab:disc}.

\begin{table}[t]\centering\small
\renewcommand{\arraystretch}{1.25}
\begin{tabularx}{\textwidth}{@{}Xccccl@{}}
\toprule
\textbf{Candidate insight} & EVoI & Addr & Conf & Cost & Verdict\\
\midrule
Long-term safety confidence is the binding objection & 2.0 & 0.80 & 0.70 & 1.0 & promote (score 1.12)\\
Subcutaneous option neutralises convenience objection & 1.4 & 0.85 & 0.65 & 0.8 & promote (0.97)\\
New-mechanism ``wait-and-see'' hesitation & 1.8 & 0.35 & 0.55 & 1.3 & route to medical (0.27)\\
Biosimilar-era value reassurance & 1.6 & 0.30 & 0.50 & 1.4 & route to access (0.17)\\
\bottomrule
\end{tabularx}
\caption{The discernment filter. The third and fourth signals carry the second- and third-highest raw value yet the two lowest scores: high EVoI, low addressability-by-content. The filter routes them out rather than forcing them into a campaign.}
\label{tab:disc}
\end{table}

The winning insight becomes a CDU: deliver a safety/mechanism priming sequence \emph{before} the (accepted) efficacy message, identified by a randomised holdout. Holding the four-touch plan constant in assets and spend, enumeration over orderings \eqref{eq:seq} yields a na\"ive ``lead-with-efficacy'' order scoring $L=0.176$ (an $8.9\%$ implied lift, missing a $10\%$ target) versus an optimal ``prime-then-decide'' order scoring $L=0.256$ ($11.8\%$, clearing it): a sequence premium of $+45.8\%$ at identical cost (Appendix~\ref{app:enum}). Readout is by difference-in-differences on a randomised, volume-matched holdout, accepted only if the lift clears the target and its confidence interval excludes zero.

\section{Empirical study: APOHA on a non-stationary Ozempic decision world}\label{sec:empirical}
Because the forgetting theory of \S\ref{sec:apoha} is the most speculative contribution, we subject it to a controlled empirical test. We operationalise APOHA as a running reasoning agent and ask whether value-aware adaptive forgetting actually improves decision quality under non-stationarity, relative to the two natural alternatives: never forgetting, and the fixed exponential half-life (a constant per-item decay) that APOHA is meant to replace.

\subsection{Hypotheses (pre-registered)}
\textbf{H1 (alignment):} adaptive forgetting yields higher decision-alignment than both baselines, with the largest margin just after a regime shift. \textbf{H2 (efficiency/cleanliness):} APOHA holds a smaller memory with higher retained signal-to-noise (more belief-mass on currently-relevant topics). \textbf{H3 (adaptation speed):} APOHA recovers faster after a shift than \emph{both} baselines. \textbf{H4 (value $=$ cost of forgetting):} beliefs rendered stale by a shift are preferentially and promptly forgotten, so per-topic staleness stays low. \textbf{H5 (no catastrophic forgetting):} surprise-driven plasticity does not erase persistently-true structure. The pre-registered \emph{value gate} was H1 $\wedge$ H2 $\wedge$ H4 without violating H5; H3 was treated as desirable but not decisive.

\subsection{Environment and protocol}
The world is ten obesity-treatment decision topics (weight-loss efficacy, cardiovascular-outcome benefit, GI tolerability, injectable barrier, supply/access, weight regain off-treatment, muscle preservation, competitor superiority, oral GLP-1 option, new indications) whose true relevance drifts slowly and jumps at five scheduled regime shifts chosen to mirror real GLP-1 market events (Appendix~\ref{app:exp}, Table~\ref{tab:shifts}). Each \emph{interaction} is one full cognitive cycle: sense a noisy signal about one topic, store it as a belief (a memory of an estimate at a time), revalue, forget/consolidate, and decide an attention allocation; the world returns a noisy payoff. A strict information barrier separates agent and evaluator: the agent forgets using only signals it could really observe --- contradiction by new evidence, redundancy, and its own relevance estimate --- while ground truth is used \emph{only} by the evaluator to score alignment. We ran 30 random seeds $\times$ 100 interactions; the three policies share identical sensing and decision code and differ only in the forgetting rule, so any difference is attributable to memory management alone.

\subsection{The adaptive forgetting rule (operationalising F2--F5)}
Each belief carries a topic, a stored estimate, a corroboration time, a consolidation score, and a contradiction signal (an exponential average of disagreement with fresh observations). Its forgetting hazard is
\begin{equation}\label{eq:hazardimpl}
\begin{split}
\mathrm{hazard}(i)=\sigma\big(&1.6\,\mathrm{contradiction}_i+0.9\,\mathrm{redundancy}_i+0.04\,\mathrm{age}_i\\
&{}-1.8\,\mathrm{relevance}_{\tau(i)}-1.5\,\mathrm{consol}_i\big)\times(0.5+1.6\,\beta_t),
\end{split}
\end{equation}
forgetting what is contradicted (stale) or redundant (reconstructable from corroborated siblings), protecting what is currently relevant or consolidated, with the outer factor the surprise-gated plasticity $\beta_t$ (F3): an EMA of prediction error that spikes when the world shifts. \emph{Never-forget} sets the hazard to zero; \emph{fixed half-life} uses a constant, value-blind hazard ($\approx$18-interaction half-life).

\subsection{Results}
\begin{figure}[t]\centering
\includegraphics[width=\textwidth]{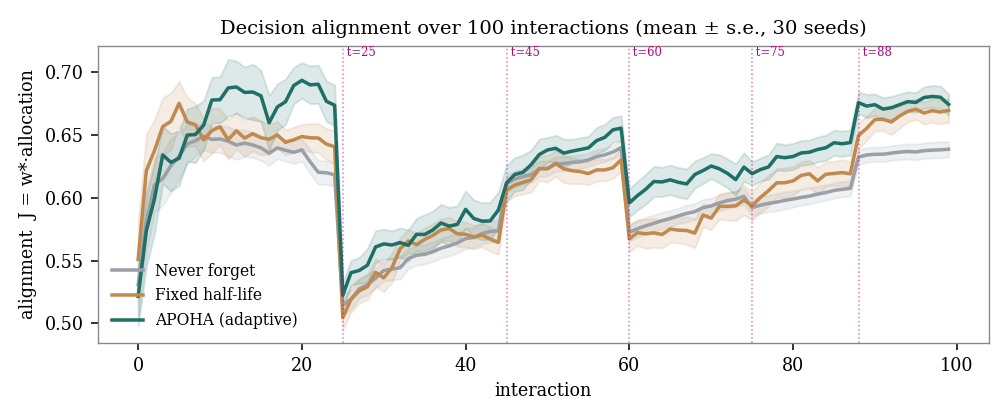}
\caption{Decision alignment across 100 interactions (mean $\pm$ s.e., 30 seeds); dotted lines mark regime shifts. APOHA (teal) leads in every regime and recovers fastest after each break.}
\label{fig:align}
\end{figure}

Table~\ref{tab:headline} gives the headline numbers. APOHA reduced cumulative regret by 24\% versus fixed-decay and 32\% versus never-forget, won decision-alignment in every regime (stable and post-shift), and held the smallest, cleanest memory.

\begin{table}[h]\centering\small\renewcommand{\arraystretch}{1.2}
\begin{tabularx}{\textwidth}{@{}lcccccc@{}}
\toprule
\textbf{Policy} & Align (all) & Align (post-shift) & Align (stable) & Cum.\ regret & Mem size & SNR\\
\midrule
Never forget & 0.6036 & 0.5949 & 0.6124 & 7.98 & 57.7 & 0.666\\
Fixed half-life & 0.6115 & 0.6008 & 0.6222 & 7.19 & 22.8 & 0.614\\
\textbf{APOHA} & \textbf{0.6290} & \textbf{0.6201} & \textbf{0.6378} & \textbf{5.44} & \textbf{11.0} & \textbf{0.734}\\
\bottomrule
\end{tabularx}
\caption{Summary over 30 seeds $\times$ 100 interactions. Higher alignment and SNR are better; lower regret and memory are better.}
\label{tab:headline}
\end{table}

\begin{figure}[h]\centering
\begin{minipage}{0.49\textwidth}\includegraphics[width=\textwidth]{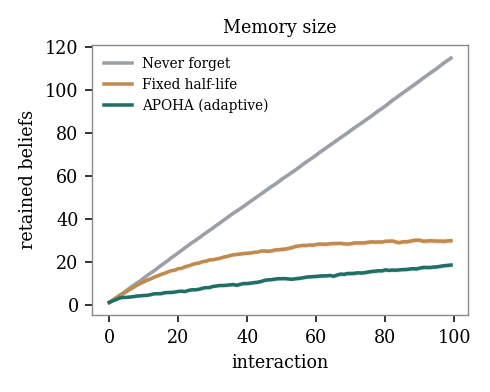}\end{minipage}\hfill
\begin{minipage}{0.49\textwidth}\includegraphics[width=\textwidth]{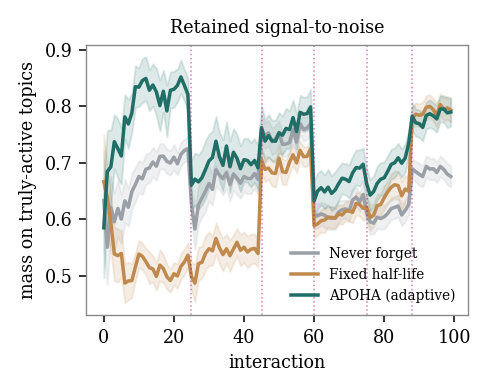}\end{minipage}
\caption{Left: memory size --- never-forget balloons toward $\sim$96 beliefs by the final third while APOHA stays near 11. Right: retained signal-to-noise --- adaptive forgetting (teal) is cleanest, and fixed-decay (amber) is \emph{worse} than never-forget, the asymmetry discussed below.}
\label{fig:memsnr}
\end{figure}

We evaluate the hypotheses. \textbf{H1 (supported):} APOHA leads on alignment overall, post-shift, and even in stable stretches (Fig.~\ref{fig:align}), the last because a cleaner memory dilutes the decision less. \textbf{H2 (supported):} $\sim$11 retained beliefs versus 23 and 58, at the highest SNR. \textbf{H3 (partially supported):} APOHA's mean recovery lag (6.38 interactions) clearly beats never-forget (8.00) but only ties fixed-decay (6.31); adaptive forgetting matches blind forgetting on raw \emph{speed}, and earns its advantage in \emph{what} it keeps, not how fast it reacts --- reported as a genuine non-dominance. \textbf{H4 (supported):} Fig.~\ref{fig:h4}; after a topic's truth shifts, APOHA forgets the now-stale beliefs and rebuilds, driving per-topic staleness back down, while never-forget lets staleness persist with up to 13 stale beliefs accumulated. \textbf{H5 (supported):} despite surprise spikes at every shift (Fig.~\ref{fig:internals}), alignment never collapses and stable-regime alignment is the highest of the three, indicating consolidation preserved persistently-true structure.

\begin{figure}[t]\centering
\includegraphics[width=\textwidth]{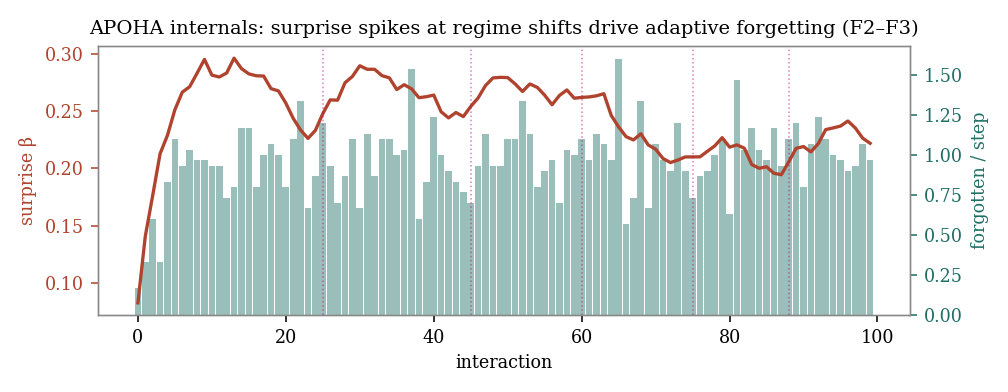}
\caption{APOHA internals. Surprise $\beta$ (red) spikes at each regime shift and drives a burst of forgetting (teal bars), then subsides into consolidation --- the surprise-gated plasticity of F3 operating as specified.}
\label{fig:internals}
\end{figure}

\begin{figure}[t]\centering
\includegraphics[width=\textwidth]{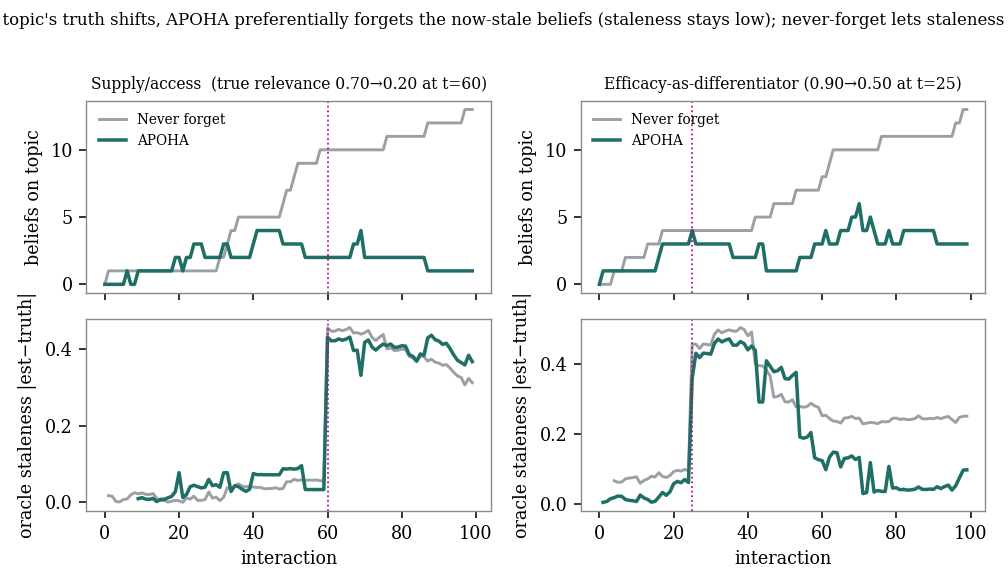}
\caption{H4 evidence. When a topic's truth shifts (dotted line), never-forget (grey) accumulates stale beliefs and lets staleness persist; APOHA (teal) forgets the now-wrong beliefs and rebuilds, driving staleness back down. Value is, operationally, the cost of forgetting.}
\label{fig:h4}
\end{figure}

\paragraph{The sharpest finding.} Fixed-decay's SNR (0.614) is \emph{lower} than never-forget's (0.666): blind forgetting discards good and bad beliefs alike and can be worse than not forgetting at all. Only \emph{value-aware} forgetting raises SNR. This asymmetry (Fig.~\ref{fig:memsnr}, right) shows the benefit is specific to the \emph{adaptive} character of \eqref{eq:hazardimpl}, not to forgetting per se --- a result that materially strengthens, and refines, the claim of \S\ref{sec:apoha}: the operator that matters is value-gated forgetting, and a naive half-life (the object APOHA replaces) is not merely weaker but can be actively harmful.

\paragraph{Higher-order learning.} Within each regime all policies' alignment climbs as beliefs accrue, then drops at a shift and re-climbs. The diagnostic of \emph{higher-order} learning is that APOHA's stable-regime alignment exceeds the others while holding far fewer beliefs: it has learned \emph{what to keep}, not merely accumulated. The flow's convergence to a small, stable, high-SNR memory across all five shifts --- never collapsing to the trivial fixed point nor oscillating --- is the empirical behaviour Conjecture~\ref{conj:attractor} predicts, and is our principal (non-proving) evidence for it.
\section{Objections and responses}\label{sec:objections}
Because the forgetting theory of \S\ref{sec:apoha} is the most speculative contribution, we subject it to a structured critique from five disciplinary vantage points and record our responses; the disagreements sharpen rather than refute the framework.

\paragraph{Information theory.} \emph{Objection:} F1 hides a circularity --- one compresses toward $Y$ whose value is itself being learned, so the target moves as one aims. \emph{Response:} the circularity is the engine, not a defect, but it obliges the contraction guarantee of Conjecture~\ref{conj:attractor}; we therefore promote that guarantee to the paper's central open problem rather than assuming convergence. Empirically (\S\ref{sec:empirical}) the joint forget/value loop converged stably across all five shifts rather than diverging, which is encouraging but not a proof.

\paragraph{Neuroscience.} \emph{Objection:} forgetting without its conjugate, replay/consolidation, produces catastrophic forgetting under a surprise spike. \emph{Response:} accepted and incorporated --- F5 makes consolidation $C$ a mandatory partner of forgetting; APOHA is an antagonist pair, not forgetting alone. This is also the conservation law the dynamical objection demands, and \S\ref{sec:empirical} confirms it operationally (H5): despite surprise spikes at every shift, stable-regime alignment remained the highest of the three policies, indicating persistently-true structure was preserved.

\paragraph{Dynamical systems.} \emph{Objection:} renormalisation flows admit a trivial ``forget everything'' fixed point and may drift to oblivion; an attractor is asserted, not shown. \emph{Response:} consolidation supplies the conservation term that excludes the trivial fixed point; the spectral gap (F6) is the diagnostic that the flow is not collapsing. Demonstrating a strictly positive gap is precisely Conjecture~\ref{conj:attractor}. Empirically (\S\ref{sec:empirical}), memory neither bled to zero nor exploded: it stabilised near eleven beliefs concentrated on the truly-active topics across all five shifts.

\paragraph{Category theory.} \emph{Objection:} treating lossless forgetting as the ideal is backwards; the valuable forgetting is the non-invertible part, since abstraction is the irreversible quotient. \emph{Response:} accepted and adopted --- F5 demotes the lossless case to mere compression and identifies lossy, non-adjoint forgetting as where higher-order value is created; the relevance spectrum measures which regime the system is in.

\paragraph{Governance.} \emph{Objection:} adaptive forgetting is an attack surface --- an adversary who makes harmful or mandated records appear redundant induces their loss. \emph{Response:} adaptive forgetting must operate only above an immutable compliance floor, every forget must be logged in an append-only audit trail, and the value-probe of Axiom~\ref{ax:value} must itself be adversarially audited. We treat these as deployment preconditions, not optional hardening.

The corrected thesis: \emph{higher-order learning is an antagonist flow of irreversible forgetting and consolidation in which value is the cost of forgetting, higher value is what survives repeated forgetting, and the system's final form is the attractor this flow carves out --- bounded below by an immutable floor.}

\section{Limitations and governance}\label{sec:limits}
The sequence premium is exact only for the small instance enumerated; at scale it depends on the fidelity of the learned belief dynamics and on solving the associated Markov decision process. The discernment scores and the value-to-objective conversion are illustrative pending field calibration. The empirical study of \S\ref{sec:empirical}, while pre-registered and run over 30 seeds, is \emph{synthetic}: its magnitudes depend on the chosen regime-shift schedule and noise, and its agent approximates value from contradiction and redundancy rather than from realised, outcome-attributed payoff --- closing that loop is the natural next step. The study also revealed a genuine non-dominance (H3): adaptive forgetting does not beat blind forgetting on raw adaptation \emph{speed}, only on what it retains. And Conjecture~\ref{conj:attractor} remains unproven: the observed stable convergence is supporting evidence, not a theorem. In regulated settings, adaptive forgetting requires an immutable retention floor, auditability of every forget, and adversarial testing of the value-probe, as argued in \S\ref{sec:objections}.

\section{Conclusion}\label{sec:conclusion}
We have argued that discernment, sequencing, leverage, and forgetting are facets of one object --- the counterfactual effect of an intervention --- and assembled them into a single calculus. Value is decision-relevance, established causally; action carries an order as well as a size; leverage is a shadow price that unifies disparate domains; and, most speculatively, forgetting is the operator by which value is learned, with higher-order value defined as the invariant that survives repeated forgetting. We moved this last claim from assertion toward evidence: in a controlled non-stationary test, value-aware adaptive forgetting cut decision-regret by roughly a quarter to a third against both never-forgetting and a fixed half-life, kept a markedly cleaner and smaller memory, and converged stably rather than collapsing --- while a blind half-life proved worse than no forgetting at all, showing the benefit is specific to \emph{value-aware} forgetting. The framework's empirical layers are therefore not merely testable but tested; its theoretical core --- the existence of a non-trivial attractor for the forget--consolidate flow --- remains an open problem for which we now offer supporting evidence. In a world where candidate insights are abundant, the scarce and decisive capability is the disciplined, value-aware removal of almost all of them.

\appendix
\section{Notation}\label{app:notation}
$z_{m,t}$ belief on need $m$ at time $t$; $\lambda$ memory; $\kappa_a$ asset quality; $g_m$ salience; $\pi_m$ priming gate; $w_m$ decision weight; $L$ decision-weighted belief; $\sigma$ schedule; $\rho,\theta,\beta$ headroom, conversion, persuasion; $\mathrm{EVoI}$ value of information; $\tau_s$ causal effect; $\eta$ shadow price; $J,M,X,Y$ objective, memory, raw past, target; $R,C$ forget-revalue and consolidation operators; $\mathcal{V}$ value function.

\section{Exact sequence-premium instance}\label{app:enum}
One persona, four touch-cards over four periods: $T_1$ = (mechanism, $\kappa{=}1.0$), $T_2,T_3$ = (efficacy, $\kappa{=}0.8,0.5$, gated), $T_4$ = (convenience, $\kappa{=}0.8$). Parameters $\lambda{=}0.6$, $\beta{=}1.0$, $K{=}0.5$; weights $w=(0.2,0.5,0.3)$; salience $g=(0.7,0.9,0.6)$; conversion $\rho{=}0.25,\theta{=}2.5$. Exhaustive evaluation of all $24$ orderings under \eqref{eq:E1}--\eqref{eq:E3E4} gives a worst order $L=0.130$, a message-priority na\"ive order $L=0.176$ (implied lift $8.9\%$), and an optimal order $L=0.256$ (implied lift $11.8\%$). The sequence premium is $\Pi_{\mathrm{seq}}=(0.256-0.176)/0.176=+45.8\%$ at identical cost; the optimal order defers the mechanism touch so its priming peaks immediately before the decisive efficacy touch and lands efficacy last for recency, confirming Proposition~\ref{prop:noncommute} with a strictly positive spread.

\section{Experimental reproducibility}\label{app:exp}
The world is ten topics whose initial true relevance vector is
\[
w_0=(0.90,\,0.30,\,0.60,\,0.50,\,0.70,\,0.30,\,0.20,\,0.30,\,0.10,\,0.20),
\]
with slow Gaussian drift ($\sigma{=}0.01$/step), observation noise $\sigma{=}0.10$, spurious-belief probability $0.15$, soft-allocation temperature $4$, and the five scheduled regime shifts of Table~\ref{tab:shifts}. The agent state is a list of belief items $\langle$topic, estimate, birth, last-corroboration, consolidation, contradiction-EMA$\rangle$; aggregation is consolidation-weighted; the forgetting hazard is \eqref{eq:hazardimpl}; consolidation increments (capped at $3$) for corroborated beliefs on currently-relevant topics; surprise $\beta_t$ is an EMA of prediction error. Baselines: never-forget (hazard $0$); fixed half-life (constant hazard, $\approx$18-interaction half-life). Alignment is $J=w^\top p$ where $p$ is the agent's allocation and $w$ the (evaluator-only) ground truth; cumulative regret sums $J_{\max}-J$. Results are means over 30 seeds $\times$ 100 interactions and are deterministic per seed. A reference implementation and the figure-generating code accompany this preprint.

\begin{table}[h]\centering\small\renewcommand{\arraystretch}{1.15}
\begin{tabularx}{\textwidth}{@{}c X l@{}}
\toprule
\textbf{Interaction} & \textbf{Regime shift (real-world analogue)} & \textbf{True-relevance change}\\
\midrule
25 & Head-to-head competitor superiority (tirzepatide) & competitor $0.30\!\to\!0.85$; efficacy-as-differentiator $0.90\!\to\!0.50$\\
45 & Cardiovascular-outcomes landmark (SELECT-type) & CV benefit $0.30\!\to\!0.85$; durability/regain $0.30\!\to\!0.60$\\
60 & Supply normalises & supply/access $0.70\!\to\!0.20$\\
75 & Oral GLP-1 launches & oral option $0.10\!\to\!0.70$; injectable barrier $0.50\!\to\!0.15$\\
88 & New indications (MASH / OSA / CKD) & new indications $0.20\!\to\!0.75$; muscle preservation $0.20\!\to\!0.55$\\
\bottomrule
\end{tabularx}
\caption{The regime-shift schedule used in \S\ref{sec:empirical}. Between shifts every topic also undergoes slow drift, so the target is never stationary. Magnitudes are illustrative and chosen to make the test demanding; they are not claims about real GLP-1 market values.}
\label{tab:shifts}
\end{table}

\section{Worked formalism}\label{app:worked}
This appendix works the three central mechanisms numerically, so the formalism is concrete and checkable. All figures are exact.

\subsection*{D.1\quad Value of information as decision-change (\protect\eqref{eq:evoi})}
A planner chooses between action $A$ (safety-priming) and $B$ (efficacy push). The world is in state ``primed-sensitive'' with probability $0.6$, else $0.4$. Payoffs (percentage-point lift):
\[
\begin{array}{l|cc}
 & \text{primed-sensitive }(0.6) & \text{not }(0.4)\\\hline
A & 12 & 4\\
B & 6 & 9
\end{array}
\]
\emph{Choosing blind} (before the signal), one compares expected payoffs:
$\Eb{A}=0.6(12)+0.4(4)=8.8$, $\Eb{B}=0.6(6)+0.4(9)=7.2$, so the best blind action is $A$ with value $8.8$.
\emph{Choosing informed} (after a signal that reveals the state), one picks the best action in each state:
$0.6\max(12,6)+0.4\max(4,9)=0.6(12)+0.4(9)=10.8$.
Hence
\[
\mathrm{EVoI}=10.8-8.8=2.0 .
\]
The value is strictly positive precisely because the information \emph{changes} the action in the second state (from $A$ to $B$); had $A$ remained optimal in both states, the two computations would coincide and $\mathrm{EVoI}=0$. This is the formal content of ``value is decision-change, not truth.''

\subsection*{D.2\quad Non-commuting touches and the sequence premium (\protect\eqref{eq:E1}--\protect\eqref{eq:E2}, Proposition~\ref{prop:noncommute})}
Two touches: $T_{\mathrm{m}}$ (mechanism, the prerequisite) and $T_{\mathrm{e}}$ (efficacy, gated). Parameters $\lambda=0.6$, $K=0.5$; $\kappa_{\mathrm m}g_{\mathrm m}=1.0(0.7)$, $\kappa_{\mathrm e}g_{\mathrm e}=0.8(0.9)$; decay is applied at the start of each period. Starting from $z=0$:

\emph{Order} $T_{\mathrm m}\!\to\!T_{\mathrm e}$. Period 1 delivers mechanism: $z_{\mathrm m}=1.0(0.7)=0.70$. Period 2 decays $z_{\mathrm m}\!\to\!0.42$, then delivers efficacy through the gate $\pi_{\mathrm e}=\tfrac{0.42}{0.42+0.5}=0.4565$: $z_{\mathrm e}=0.8(0.9)(0.4565)=0.3287$. \textbf{Terminal} $(z_{\mathrm m},z_{\mathrm e})=(0.42,\,0.3287)$.

\emph{Order} $T_{\mathrm e}\!\to\!T_{\mathrm m}$. Period 1 delivers efficacy with no priming, $\pi_{\mathrm e}=\tfrac{0}{0+0.5}=0$, so $z_{\mathrm e}=0$. Period 2 delivers mechanism: $z_{\mathrm m}=0.70$. \textbf{Terminal} $(z_{\mathrm m},z_{\mathrm e})=(0.70,\,0)$.

The terminal states differ, so $T_{\mathrm e}\!\circ\!T_{\mathrm m}\neq T_{\mathrm m}\!\circ\!T_{\mathrm e}$ (Proposition~\ref{prop:noncommute}). With decision weights $w_{\mathrm m}=0.3,\,w_{\mathrm e}=0.7$ and $\beta=1$ in $L=\sum_m w_m(1-e^{-\beta z_m})$,
\[
L(T_{\mathrm m}\!\to\!T_{\mathrm e})=0.2990,\qquad L(T_{\mathrm e}\!\to\!T_{\mathrm m})=0.1510,\qquad \Pi_{\mathrm{seq}}=\frac{0.2990-0.1510}{0.1510}=+98.0\% .
\]
Identical content and cost; the order alone nearly doubles the decision-weighted belief, because firing efficacy before priming wastes it on a closed gate.

\subsection*{D.3\quad Value as the cost of forgetting (\protect\eqref{eq:valforget})}
Three topics with true relevance $w^{*}=(0.9,\,0.2,\,0.5)$. The objective is
\[
J(M)=w^{*\top}\,\mathrm{softmax}\!\big(4\,\hat w(M)\big),
\]
where $\hat w(M)$ is the consolidation-weighted topic aggregate of the retained beliefs. Memory $M$ holds three beliefs:
\[
\begin{aligned}
b_1&=\langle\text{topic }1,\ 0.90\rangle && (\text{accurate, unique}),\\
b_2&=\langle\text{topic }1,\ 0.85\rangle && (\text{redundant}),\\
b_3&=\langle\text{topic }2,\ 0.80\rangle && (\text{wrong: truth }0.2).
\end{aligned}
\]
Then $J(M)=0.6004$, and the cost of forgetting each belief, $\val(i)=J(M)-J(M\setminus i)$, is
\[
\begin{array}{lccl}
\text{belief} & J(M\setminus i) & \val(i) & \text{reading}\\\hline
b_1\ \text{(accurate, unique)} & 0.5834 & +0.0170 & \text{costly to lose} \Rightarrow \text{keep}\\
b_2\ \text{(redundant)} & 0.6172 & -0.0168 & \text{removal \emph{helps}} \Rightarrow \text{forget}\\
b_3\ \text{(wrong, off-topic)} & 0.8687 & -0.2683 & \text{strongly harmful} \Rightarrow \text{forget first}
\end{array}
\]
The signs carry the whole idea. Forgetting $b_1$ has a positive cost (alignment falls), so $b_1$ is valuable. Forgetting the redundant $b_2$ has \emph{negative} cost --- the system is better off without it --- and forgetting the stale, wrong $b_3$ is strongly beneficial. ``Value is the cost of forgetting'' is therefore literal: the items worth keeping are exactly those expensive to lose, and the adaptive hazard \eqref{eq:hazardimpl} recovers the same ordering from \emph{available} signals alone --- $b_3$ has the highest contradiction (its stored $0.80$ clashes with fresh evidence near $0.20$), $b_2$ the highest redundancy (a corroborated sibling on topic~1), and $b_1$ is protected by topic relevance and uniqueness --- so the operational rule forgets $b_3$, then $b_2$, and keeps $b_1$, matching the oracle ranking $\val(b_3)<\val(b_2)<\val(b_1)$.

\subsection*{D.4\quad Higher-order value as a renormalisation invariant (\protect\eqref{eq:rg})}
Linearise the forget-and-revalue map $R$ about its fixed point and diagonalise it. Suppose two directions in value-space with multipliers $\mu_{\mathrm{rel}}=1.3$ and $\mu_{\mathrm{irr}}=0.5$. Iterating $R$ scales a component along each direction by its multiplier, so after $k$ rounds of coarse-graining the irrelevant component is suppressed as $0.5^{k}\!\to\!0$ while the relevant component persists ($1.3^{k}$, bounded by the consolidation floor). The \emph{spectral gap} $\mu_{\mathrm{rel}}/\mu_{\mathrm{irr}}=2.6>1$ certifies a clean separation. Higher-order value is, by definition, the projection onto the relevant eigenspace --- the structure that survives repeated forgetting --- and noise is what the flow destroys. A vanishing gap ($\mu_{\mathrm{rel}}\!\approx\!\mu_{\mathrm{irr}}$) would signal that the system can no longer tell signal from noise; the empirical relevance-spectrum diagnostic (F6) monitors exactly this ratio.


\end{document}